\documentclass[a4paper,USenglish,cleveref, autoref ]{lipics-v2021}

\pdfoutput=1 
\hideLIPIcs  

\usepackage{enumitem}
\usepackage{stmaryrd}

\title{Recognizable Picture Languages: Separating UREC from coUREC via Communication Complexity} 

\titlerunning{Separating UREC from coUREC} 

\author{Antonin Callard}{LIRMM -- CNRS and Université de Montpellier, France}{contact@acallard.net}{}{}

\author{ Andrei Romashchenko}{LIRMM -- CNRS and Université de Montpellier, France}{andrei.romashchenko@lirmm.fr}{}{supported in part by the ANR project FLITTLA (ANR-21-CE48-0023)}

\author{V\'eronique Terrier}{Université Caen Normandie, ENSICAEN, CNRS, Normandie Univ, GREYC; F-14000 Caen, France}{veronique.terrier@unicaen.fr}{}{}

\author{Pascal Vanier}{Université Caen Normandie, ENSICAEN, CNRS, Normandie Univ, GREYC; F-14000 Caen, France}{pascal.vanier@unicaen.fr}{}{}

\authorrunning{A.~Callard, A.~Romashchenko, V.~Terrier, P.~Vanier} 



\acknowledgements{The authors thank Sophie Laplante and Artur Riazanov for their helpful comments on the history of the technique of lifting from query complexity to communication complexity.} 

\nolinenumbers 

\newcommand{\details}[1]{\captionsetup{font=footnotesize}\caption*{#1}}
\newcommand{\bigO}{\mathcal{O}}

\newcommand{\rec}{\mathrm{REC}}
\newcommand{\corec}{\mathrm{coREC}}
\newcommand{\urec}{\mathrm{UREC}}
\newcommand{\courec}{\mathrm{coUREC}}
\newcommand{\loc}{\mathrm{Local}}

\newcommand{\alphabet}[1][A]{\mathcal{#1}}

\RenewDocumentCommand{\square}{O{}}{\begin{tikzpicture}[x=1.4ex,y=1.4ex,baseline=.15ex]
    \ifstrempty{#1}
    {\draw[black!30] (0,0) rectangle (1,1);}
    {\draw[black!30, fill=#1] (0,0) rectangle(1,1);}
  \end{tikzpicture}}

\newcommand{\ncert}{\mathrm{N}^{cert}} 
\newcommand{\concert}{\mathrm{coN}^{cert}} 
\newcommand{\ucert}{\mathrm{UN}^{cert}} 
\newcommand{\coucert}{\mathrm{coUN}^{cert}} 

\newcommand{\ncc}{\mathrm{N}^{cc}} 
\newcommand{\concc}{\mathrm{coN}^{cc}}  
\newcommand{\ucc}{\mathrm{UN}^{cc}} 
\newcommand{\coucc}{\mathrm{coUN}^{cc}}  

\usepackage[dvipsnames]{xcolor}
\usepackage{xspace}
\usepackage{float}

\usepackage{tikz}
\usepackage{pgfplots}

\usetikzlibrary{
  positioning,
  arrows.meta,
  calc,
  matrix,
  arrows,
  shapes,
  patterns,
  automata,
  decorations.pathreplacing,
  decorations.pathmorphing,
  decorations.markings,
  fit,
  shapes.geometric,
  patterns,
  patterns.meta
}

\usetikzlibrary{patterns}

\definecolor{darkgreen}{rgb}{0.,0.4,0.}
\definecolor{darkblue}{rgb}{0.,0.,0.4}
\definecolor{darkorange}{rgb}{0.7,0.4,0}
\colorlet{TilingGrid}{black!15!white}

\newtheorem{question}{Question}

\colorlet{vertcolor}{Turquoise}
\colorlet{horcolor}{Mulberry}
\colorlet{diagcolor}{LimeGreen}
\colorlet{nodescolor}{Red}
\definecolor{midcolor}{HTML}{5c4033}  
\definecolor{countkcolor}{HTML}{6c40ef}
\colorlet{countkcolor2}{countkcolor!70}
\definecolor{countbcolor}{HTML}{0d2de2}
\definecolor{innerproductcolor}{HTML}{ba341c} 
\colorlet{unarycountercolor}{orange}

\usepackage{tikz}

\newcommand{\crossbox}{%
  \tikz[baseline=0ex, scale=0.23]{%
    \path[fill=nodescolor] (0,0) rectangle (1,1);
    \fill[black] (.5,.5) circle[radius=.15];
  }%
}
\newcommand{\hatchedbox}[3]{%
  \tikz[baseline=0ex, scale=0.23]{%
    \path[
      draw=midcolor,
      line width=1.0pt, 
      pattern={Lines[
        angle=#1,
        line width=#2,
        distance=#3
      ]},
      pattern color=midcolor
    ] (0,0) rectangle (1,1);
  }%
}
\newcommand{\brownboxslash}{\hatchedbox{45}{0.30pt}{1.5pt}} %
\newcommand{\brownboxbackslash}{\hatchedbox{135}{0.3pt}{1.5pt}}
\newcommand{\unarymarkerbox}{%
  \tikz[baseline=0ex, scale=0.23]{%
    \draw[unarycountercolor,thick] (0,0) rectangle (1,1);
    \draw[unarycountercolor,thick] (0,0) -- (1,1) (0,1) -- (1,0);
  }%
}

\NewDocumentCommand{\larrowsquare}{}{\begin{tikzpicture}[x=1.8ex,y=1.8ex,baseline=0ex]
    \draw[midcolor,thick] (0,0) rectangle (1,1);
      \draw[-{Stealth[length=1mm,width=1mm]}, midcolor!80]
           (.5,.5) -- ++(0.3,0); 
  \end{tikzpicture}}  

\NewDocumentCommand{\rarrowsquare}{}{\begin{tikzpicture}[x=1.8ex,y=1.8ex,baseline=0ex]
    \draw[midcolor,thick] (0,0) rectangle (1,1);
      \draw[-{Stealth[length=1mm,width=1mm]}, midcolor!80]
           (.5,.5) -- ++(-0.3,0); 
  \end{tikzpicture}}

\begin{document}

\maketitle

\begin{abstract}
We introduce communication-complexity lifting techniques into the study of recognizable picture languages.
As an application, we resolve a long-standing open problem of Anselmo et al. (2006) by constructing a language in UREC whose complement does not belong to REC.
Our proof of the lower bound is inspired by the communication-complexity approach that Göös et al. (2022) used for unambiguous automata, 
although the setting of picture languages requires substantially different technical tools.
\end{abstract}

\newpage

\section{Introduction}

In this work, we study nondeterminism in the theory of recognizable picture languages. In particular, we consider  the implications of restricting nondeterminism to unambiguous nondeterminism.

Generally speaking,
nondeterminism can be seen as a form of exhaustive search in which several computational branches are considered in parallel. A fundamental question in theoretical computer science is whether such a parallel exhaustive search can be replaced by an efficient deterministic computation. This question appears in many models of computation and has been widely studied in complexity theory.

Besides the total number of possible computational branches, another important parameter when talking about nondeterminism is the number of \emph{accepting} branches. Restricting this number to at most one branch leads to the notion of unambiguous computation and to the complexity class $\mathrm{UP}$, which is a natural refinement of $\mathrm{NP}$. The importance of unambiguous computations was illustrated, for example, by the famous Valiant--Vazirani theorem~\cite{valvaz86}, which gives a randomized reduction from $\mathrm{NP}$ to problems with a unique witness.

Similar questions appear in formals languages and automata theory (see \cite{goldstine2002descriptional} for a survey).
The classical equivalence between deterministic and nondeterministic finite automata shows that nondeterminism does not increase the expressive power of finite automata, but may significantly reduce their number of states.
The difference between general nondeterministic and unambiguous nondeterministic finite automata is a classical direction of research, see, \textit{e.g.}, \cite{lupanov1963comparison,schmidt1978succinctness,stearns1985equivalence,leung2006structurally}.
Recently, the power of unambiguous nondeterminism in finite automata has again received attention, in particular in~\cite{raskin2018icalp,widdershoven2022unambiguous,goos2022lower}.

In this paper, we study similar questions in the context of picture languages. More precisely, the class of \emph{recognizable} picture languages is a natural two-dimensional analogue of regular languages of finite words. In this setting, the classes $\rec$ and $\urec$ respectively refer to the nondeterministic and unambiguously nondeterministic recognizable picture languages, and are thus similar to (word) languages recognized by nondeterministic and unambiguous finite automata. Although the class $\rec$ and its variants have been studied extensively, the structural consequences of restricting nondeterminism to the unambiguous case with the $\urec$ class are still not well understood. Our goal is to study the relationship between the classes $\rec/\corec$ and $\urec/\courec$ and to clarify the role of unambiguity in the theory of two-dimensional words/pictures.

\subsection{Recognizable picture languages: the formal definitions}

We now recall the formal definitions of recognizable picture languages and their unambiguous variants, as introduced by Giammarresi \& Restivo in~\cite{giammarresi1992recognizable}. This definition is based upon a notion of locally valid tilings.

\begin{definition}[picture]
A \emph{picture} of size $m \times n$ over an alphabet $\Sigma$ is a rectangular pattern $p \in \Sigma^{m \times n}$.
A \emph{picture language} is a subset of $\Sigma^{**}$, where $\Sigma^{\ast\ast}$ is the set of all pictures over~$\Sigma$.

Let $\#$ be a special symbol not in $\Sigma$. The \emph{bordered picture} of $p$ is the picture $\hat{p}$  of size $(m+2) \times (n+2)$ obtained by surrounding the picture $p$ with $\#$ symbols.
\end{definition}

\begin{definition}[patches, valid picture]
A \emph{patch} is a picture of size $2 \times 2$. Given an alphabet~$\Sigma$ and a set of patches $\Theta$ over $\Sigma\cup \{\#\}$, a picture $p$ is $\Theta$-\emph{valid} if all the subpictures of size~${2 \times 2}$ in the bordered picture $\hat{p}$ belong to~$\Theta$.
\end{definition}
A language that can be described by the set of its $2 \times 2$-subpictures is called \emph{local}:
\begin{definition}[local picture language]
\label{def:local_language}
  A picture language $L$ over $\Sigma$ is \emph{local} if there exists a set of patches $\Theta$ over $\Sigma\cup\{\#\}$ such that $L$ is exactly the set of $\Theta$-valid pictures.
 In this case we write $L=\loc(\Theta)$.
\end{definition}

Among their many characterizations, regular languages of finite words are exactly the morphic images of local languages. Giammarresi and Restivo extended this notion to picture languages by introducing in~\cite{giammarresi1992recognizable}
the class $\rec$ of pictures languages \emph{recognized} as projections of local tiling systems:
\begin{definition}[projection]
For $\Gamma$ and $\Sigma$ two alphabets and $\pi:\Gamma\rightarrow\Sigma$ a map, the projection of a picture $p\in\Gamma^{m \times n}$ is the picture $\pi(p)\in\Sigma^{m \times n}$ given by $\pi(p)(i,j)=\pi(p(i,j))$. \\
The projection $\pi(L)$ of a language $L$ over $\Gamma$ is the language over $\Sigma$ defined as $\{\pi(p):p\in L\}$.
\end{definition}

\begin{definition}[tiling system]
\label{def:tiling_system}
A tiling system is  a quadruple $\mathcal{T}=(\Sigma,\Gamma,\Theta,\pi)$ where
\begin{itemize}
\item $\Sigma$ and $\Gamma$ are finite alphabets;
\item $\Theta$ is a set of patches over $\Gamma\cup \{\#\}$;
\item $\pi:\Gamma\rightarrow\Sigma$ is a an alphabetic map.
\end{itemize}
The language $L$ recognized by the tiling system $\mathcal{T}$ is the projection of the local language $\loc(\Theta)$,
that is, $L=\pi(\loc(\Theta))$. A language is set to be \emph{recognizable} if it is recognized by some tiling system $\mathcal{T}$.
\end{definition}

\begin{definition}[$\rec$]
  \label{def:rec}
  The family of recognizable languages is denoted by~$\rec$, and the family of complements of recognizable languages by $\corec$.
\end{definition}

\begin{remark}
The choice of sizes $2 \times 2$ for patches in this serie of definitions is rather arbitrary.
Equivalently, $\rec$ can be defined as projections of languages specified by local constraints in terms of \emph{dominos} (\textit{i.e.}, patches of size $2 \times 1$ and $1 \times 2$), see \cite{giammarresi1997two}.
\end{remark}

A tiling system is called unambiguous if every picture has a unique preimage in its corresponding local language.
\begin{definition}[$\urec$]
  A tiling system $\mathcal{T}=(\Sigma,\Gamma,\Theta,\pi)$ recognizing a language $L$ is \emph{unambiguous} if every picture $p\in L$ admits a unique preimage $q\in\loc(\Theta)$ such that~${p=\pi(q)}$.
  A picture language that is recognized by an unambiguous tiling system is called \emph{unambiguous}.\\
$\urec$  denotes the family of all unambiguous recognizable picture languages, and $\courec$ the family of their complements.
\end{definition}

\subsection{Brief historical survey}

Although various two-dimensional automata and grammars had been studied much earlier,
it was in the 1990s that \emph{recognizable picture languages} emerged as a coherent research area.
In the 1990s, Giammarresi and Restivo introduced the class $\rec$ and laid the foundations of the theory of recognizable picture languages
\cite{giammarresi1992recognizable,giammarresi1997two}.
Inoue and Takanami connected this framework with two-dimensional automata \cite{inoue1994characterization}.
Giammarresi, Restivo, Seibert, and Thomas provided a logical characterization of recognizability in terms of
monadic second-order logic \cite{giammarresi1996monadic}.
Thus, the same class received equivalent descriptions in terms of local tilings, automata, and logic.
Latteux and Simplot established connections between recognizable
picture languages and classical one-dimensional context-sensitive languages
\cite{latteux1997context}, and Matz developed logical and algebraic methods
for obtaining negative results, in particular for proving that certain picture
languages are not recognizable \cite{matz1998piecewise}.

In the 2000s, Anselmo, Giammarresi, Madonia, and Restivo studied determinism and ambiguity in $\rec$ \cite{anselmo2006unambiguous}.
Their further results concerned deterministic and unambiguous subclasses of $\rec$, as well as necessary conditions for recognizability \cite{anselmo2010deterministic,anselmo2010classes}.
We omit several other results that are less directly related to the questions considered in this paper.

These works established $\rec$ as a natural two-dimensional counterpart of the class of regular languages.
However,  they also revealed that the relationships between recognizability, determinism, ambiguity, and complementation
are more subtle for picture languages than for word languages \cite{giammarresi1997two,anselmo2006unambiguous}.

\subsection{Major structural questions}

Concerning the relationships between the classes $\rec$, $\corec$, $\urec$, and $\courec$, let us recall two important facts:
\begin{itemize}
\item $\rec$ is not closed under complement, \textit{i.e.},
      $\rec\neq\corec$; see~\cite{szepietowski92ota};
\item there exist inherently ambiguous picture languages in $\rec$, \textit{i.e.},
      $\urec\subsetneq\rec$; see~\cite{anselmo2006unambiguous}.
\end{itemize}
These results left open two major questions.
\begin{question}[1994, \cite{potthoff1994nondeterminism,giammarresiRMatrix08}]
Is there a language in
\( (\rec\cap\corec)\setminus(\urec\cup\courec)\,? \)
\end{question}

\begin{question}[2006, \cite{anselmo2006unambiguous}]
Is $\urec$ closed under complement? Equivalently, is there a picture language in
\( \urec\setminus\courec \)\,?
\end{question}

\noindent
Question~1 was answered positively by Véronique Terrier in 2019
\cite{terrier2019communication}.
In this paper, we answer Question~2 by showing that $\urec$ is not closed
under complement. In fact, we prove the stronger statement:
\begin{theorem}
\label{th:main}
There is a picture language in $\urec\setminus \corec$.
\end{theorem}

Similarly to~\cite{terrier2019communication}, which answered Question~1, our proof of Theorem~\ref{th:main} uses methods from communication complexity: we realizes a communication problem with a suitable communication complexity as a specifically constructed picture language. We discuss this technique in more detail in Section~\ref{sec:main-proof}.

More specifically, for the proof of Theorem~\ref{th:main}, we use a communication problem with the following property:
the unambiguous communication complexity of computing a certain predicate should be \emph{rather low}, while the communication complexity of computing the negation of this predicate should be \emph{rather high}
(in the proof, the words  ``rather high'' and ``rather low''  will get a precise meaning).
It turns out that classical methods of communication complexity, in particular the matrix rank method, do not allow us to achieve this goal.
The issue is that the ranks of a Boolean matrix and of its complement cannot differ significantly.
Therefore, to estimate the communication complexity, we use a different approach --- the lifting technique introduced in \cite{razM1999NChierarchy, goosPW2018lift}.

Lower bounds for communication complexity obtained by lifting method have been used, for similar reasons, in the study of unambiguous finite automata, see~\cite{goos2022lower}.
However, our application of this method involves a technical feature\footnote{%
The lifting technique provides a transformation of certificate problem for a function $f(x_1,\ldots,n)$ into a communication complexity problem.
This transformation involves a \emph{gadget}
(a function with particular properties of a randomness extractor; in practice one can take the inner product).
In usual applications of the lifting technique, the length of the gadget's input is logarithmic in $n$;
in our construction, we take a gadget with much longer inputs, see
 Remarks~\ref{rem:log-gadgets} and \ref{rem:log-gadgets-2} and in Section~\ref{sec:construction},
and Section~\ref{sec:comm-compl} for a detailed discussion.%
}  which, to the best of our knowledge, has not appeared in previous applications of the method.

\medskip
The rest of the paper is organized as follows. In Section~\ref{sec:2}, we present the necessary tools from communication complexity, including the lifting technique.
In Section~\ref{sec:construction}, we describe the construction of a language in \(\urec\setminus\courec\). In Section~\ref{sec:proof-negative-claim},
we prove that the constructed language does not belong to \(\courec\) by using techniques from communication complexity.
In Section~\ref{sec:proof-positive-claim-short} we prove  that the constructed language belongs to $\urec$.

\section{Technical tools}
\label{sec:2}

We first introduce two notions of complexity for Boolean functions. The \emph{certificate complexity} measures how many bits a Boolean function needs to read from its input; and the \emph{communication complexity} measures the amount of communication needed to compute a Boolean function whose input is distributed between two (or more) participants. These two notions of complexity are related: the \emph{lifting technique} mentioned above is a method that transforms a certificate problem into a communication problem of similar complexity.

\subsection{Certificate complexity}
In the classical \emph{decision tree} model of computation (see, \textit{e.g.}, \cite{buhrman2002complexity}), the measure of complexity of a Boolean function in $n$ variables $f(x_1,\ldots, x_n)$ is the minimum number of input variables that need to be assigned a value in order to definitely establish the value of $f$.
This notion is closely connected with the idea of \emph{certificate}, which is a partial assignment of Boolean values to the variables $(x_1,\ldots, x_n)$ that determines the value of a Boolean function.

\begin{example}
For the function called \emph{majority} of $n$ Boolean arguments
\[
\mathrm{majority}_{n}(x_1,\ldots,x_n) =
\left\{
\begin{array}{rcl}
1,& \text{if $\ge n/2$ values $x_i$ are equal to $1$},\\
0,& \text{if $<n/2$ values $x_i$ are equal to $1$}
\end{array}
\right.
\]
in case $n=7$,
a partial assignment $(1,0,1,*,1,1,*)$ implies that the value of \emph{majority} is equal to $1$,
whatever values we substitute on the place of $*$'s.
Thus, this partial assignment ``certifies'' that $\mathrm{majority}_7(x_1,\ldots,x_7)$ is equal to $1$.
\end{example}
The \emph{complexity} of a certificate is the number of assigned variables: in the previous example, the complexity of $(1,0,1,*,1,1,*)$ is equal to $5$. Furthermore, the \emph{certificate complexity} of a function $f$ is the minimal integer $k$ such that one can find a family of certificates of complexity $\le k$ covering exactly the assignments on which $f$ evaluates to true\footnote{%
Technically, the formal definition that we discuss in this section corresponds to the complexity of \emph{nondeterministic} decision trees.}.

\medskip
Let us proceed with formal definitions. We say that a \emph{clause} is a finite conjunction of literals, where a literal is either a propositional variable or its negation.
The \emph{width} of a clause is the number of literals it contains.
Every clause of width $k$ can be understood as a representation of a Boolean function (in variables $x_1,\ldots,x_n$, where $n\ge k$).

A \emph{DNF formula} is a disjunction of such clauses.
Let $D = C_1 \vee \ldots \vee C_m$ be an $n$-variable DNF formula.
The \emph{width} of this DNF formula is the maximum of the widths of the clauses~$C_i$ for $i=1,\ldots,m$.
This DNF is called \emph{unambiguous} if, for every input $(x_1,\ldots, x_n) \in \{0, 1\}^n$, at most one of the clauses $C_i$ evaluates to true, \textit{i.e.}~$C_i(x_1,\ldots,x_n) = 1$.

\pagebreak
The \emph{certificate complexities} are then defined as follows:
\begin{definition}
  For any Boolean function $f:\{0,1\}^n \to \{0,1\}$, the \emph{nondeterministic certificate complexity} $\ncert(f)$ is the smallest integer $k$ such that $f$ can be written as a DNF of width $k$; and its \emph{unambiguous certificate complexity} $\ucert(f)$ is the smallest integer $k$ such that $f$ can be written as an unambiguous DNF of width $k$.

  Besides, we write $\concert(f)$ for $\ncert(\neg f)$ and $\coucert(f)$ for $\ucert(\neg f)$.
\end{definition}

\begin{example} 
For the disjunction function
\[
\mathrm{OR}_{n}(x_1,\ldots, x_n)
= \left\{
\begin{array}{rl}
1,&\text{if $x_i=1$ for at least one }i \\
0,&\text{otherwise}
\end{array}
\right.
\]
the definition implies
$\ncert(\mathrm{OR}_{n}) = 1$ (it is enough to guess one variable $x_i$ such that $x_i$=1) 
and $\ucert(\mathrm{OR}_{n}) = n$ (we have to ask the values of all variables).
\end{example}
We refer the reader to \cite{buhrman2002complexity,Kothari2016,ben2017low} for a more detailed discussion of the definition of certificate complexity and further examples.

\bigskip
We finally recall a family of results from Göös et al.\ proving that unambiguous certificate complexity and co-certificate complexity can be uniformly bounded away from one another:
\begin{theorem}[\cite{goos2015lower,balodis2023unambiguous}]
  \label{th:goos}
There exist real constants $\alpha<\beta$, an infinite family of indices ${\cal K}\subset\mathbb{N}$ and some associated functions
\( f_k:\{0,1\}^{k}\to\{0,1\} \)
for $k \in {\cal K}$ such that
\( \ucert(f_k)\;\le\;{O}(k^\alpha) \) and \( \concert(f_k)\;\ge\;\Omega(k^{\beta}) \).
Moreover, the family $\{f_k\}$ is explicit: there exists a deterministic
algorithm that, given $k$ and $x\in\{0,1\}^{k}$, computes $f_k(x)$
in time $\mathrm{poly}(k)$.
\end{theorem}

\begin{remark}
From the proof of Theorem~\ref{th:goos} it follows that the set ${\cal K}$ is polynomial-time computable and, moreover,
there is an explicitly given polynomial-time computable function $t(k)$ that is a an upper bound for $\ucert(f_k)$ and a lower bound for $\concert(f_k)$, \textit{i.e.},
such that we have
\( \ucert(f_k) \ll t(k) \ll \concert(f_k) \)
for $k \in {\cal K}$.
\end{remark}

\begin{remark}
Different concrete instantiations of this theorem (with different constructions of $f_k$ and different parameters $\alpha,\beta$) are given in a  series of works
\cite{goos2015lower,goos2016rectangles,goosPW2018lift,balodis2023unambiguous}.
In \cite{goos2015lower}, one obtains a separation with some constant $\alpha>1$ via a complicated recursive construction.
In \cite{balodis2023unambiguous}, one achieves a near-quadratic separation, \textit{i.e.}, $\beta/\alpha$ can be made arbitrarily close to $2$.
\end{remark}

\subsection{Communication complexity}
\label{sec:comm-compl}

The theory of communication complexity studies how much information two (or more) parties need to exchange in order to compute together a function when each party is initially given only part of the input.
There are several settings in which the parties act deterministically, nondeterministically, or use randomness; for comprehensive surveys, we refer to \cite{kushilevitz1997communication,rao2020communication}.

In this paper, we study the \emph{nondeterministic} and \emph{unambiguous nondeterministic} versions of communication complexity as defined below:

\begin{definition}[protocols]
 For a given function
 \( {f : X \times Y \to \{0,1\} } \)
 a \emph{nondeterministic communication protocol of cost} $c$ is a pair of functions
 \[
 A : X \times \{0,1\}^c \to \{0,1\}
 \quad \text{and} \quad 
 B : Y \times \{0,1\}^c \to \{0,1\} 
 \]
 such that: for all $(x,y) \in X\times Y$, we have $f(x,y)=1$ if and only if there exists some \emph{certificate} $z\in \{0,1\}^c$ satisfying
 \(  A(x,z)=1 \) and \(  B(y,z)=1 \).

Such a protocol is called \emph{unambiguous} if for each tuple  $(x,y) \in X\times Y$ such that $f (x, y) = 1$ there is exactly one $z$ such that $A(x,z) = 1$ and $B(y,z) = 1$.
\end{definition}
Informally, a nondeterministic communication protocol can be understood as a game of three parties: Alice (who is given an input $x\in X$) and Bob (who is given an input $y\in Y$), who aim at checking if $f(x,y) = 1$; and an omniscient but untrustworthy referee (who can access both inputs $x$ and $y$).
To proceed, the referee broadcasts a $k$-bits message~$z$, which is meant to be a certificate proving $f(x,y)=1$. If Alice and Bob confirm that $z$ is compatible with their respectively input parts $x$ and $y$, one can be sure that indeed $f(x,y) = 1$.
The \emph{cost} of the protocol is the length of the referee's message (we do not count Alice's and Bob's confirmations).

For every $f$ there is a trivial protocol: the referee broadcasts the entire pair $(x,y)$, so that Alice and Bob learn each other's part of the input, and this is obviously enough to compute $f(x,y)$. However, some functions $f$ admit much more efficient protocols, like in the following example:
\begin{example}
Denote
$
\mathrm{NEQ}_{n} : \{0,1\}^n \times  \{0,1\}^n \to  \{0,1\}
$
the non-equality predicate
\[
\mathrm{NEQ}_{n}(x_1\ldots x_n, y_1\ldots y_n)
= \left\{
\begin{array}{rl}
1,&\text{if there exists $i$ such that}\ x_i\not=y_i, \\
0,&\text{if $x_i=y_i$ for all $i$}.
\end{array}
\right.
\]
Let us define the following protocol: the referee announces the binary expansion of an index $i\in \{1,\ldots,n\}$ and a Boolean value $b$, which is meant to be a position $i$ such that $x_i \neq y_i$, and $b = x_i$ is the value of Alice's input at the $i$-th position. If Alice confirms that $x_i = b$ and Bob confirms that $y_i \neq b$, then Alice's and Bob's input parts indeed differ.

The cost of this protocol is $\lceil \log n \rceil +1$. However, it is \emph{not} unambiguous: if $x$ and~$y$  differ in several positions, the referee can broadcast several indices certifying that $x\not=y$.
\end{example}

\begin{remark}
Let
\( f : X \times Y \to \{0,1\} \)
be a Boolean function.  Denote by $M^f$ the matrix with $|X|$ rows and $|Y|$ columns such that $M^f_{x,y} = f(x,y)$ for all $x\in X$ and $y\in Y$
(the rows and columns of the matrix are indexed by elements from $X$ and $Y$ respectively).
Let $A\subset X$ and $B\subset Y$ be such that  $M^f_{x,y}  =1$ for all $(x,y)\in A\times B$.
Then the minor of $M^f$ formed by the intersection of rows from $A$ and columns from $B$ is called a \emph{combinatorial $1$-rectangle}.

Observe that  nondeterministic communication protocol of cost $k$ for $f$ induces a \emph{cover} of all $1$-entries of $M^f$ by at most $2^k$ combinatorial $1$-rectangles.
Similarly,  an unambiguous  nondeterministic communication protocol of cost $k$  induces  a \emph{partition} of the $1$-entries of $M^f$ into at most $2^k$ disjoint combinatorial $1$-rectangles.
\end{remark}

\begin{definition}[communication complexity]
Let $f : X \times Y \to \{0,1\}$ be a function with binary values.
The \emph{nondeterministic communication complexity} of this function is the minimal cost~$k$ over all nondeterministic protocols for $f$. We denote it $\ncc(f)$. 

The \emph{unambiguous (nondeterministic) communication complexity} of $f$ is the minimal cost $k$ over all
unambiguous nondeterministic protocols for $f$. We denote it $\ucc(f)$. We use the notation $\concc(f)$ and $\coucc(f)$ for $\ncc(\neg f)$ and $\ucc(\neg f)$ respectively.
\end{definition}
\begin{remark}
  Equivalently, the nondeterministic communication complexity of $f$ can be defined as logarithm of the \emph{cover number}, \textit{i.e.}, of the minimum size of a cover of the $1$\nobreakdash-entries of $M^f$
with combinatorial $1$-rectangles.
Similarly,  the unambiguous nondeterministic communication complexity of $f$ can be defined as logarithm of the \emph{partition number}, \textit{i.e.}, of the minimum size of a partition of the $1$-entries of $M^f$ into disjoint combinatorial $1$-rectangles.
\end{remark}

\begin{example}
Denote
$
\mathrm{EQ}_{n} : \{0,1\}^n \times  \{0,1\}^n \to  \{0,1\}
$
the equality predicate  defined as
\[
\mathrm{EQ}_{n}(x, y)
= \left\{
\begin{array}{rl}
1,&\text{if}\ x=y, \\
0,&\text{otherwise}.
\end{array}
\right.
\]
For this Boolean function 
\(
\ncc(\mathrm{EQ}_{n}) = n
\)
and
\(
\concc(\mathrm{EQ}_{n}) =  \log n +1,
\)
see, \textit{e.g.}, \cite{kushilevitz1997communication}.
\end{example}

The lifting method with the inner-product gadget from {\cite[theorem~4]{goos2015lower} and \cite{goos2016rectangles}} (see also \cite[theorem~7]{goos2022lower}) allows to transfer certificate complexity results in the domain of communication complexity:

\begin{theorem}[Lifting with the inner-products]
\label{th:bounds_for_cc}
Let
\( f:\{0,1\}^{k}\to\{0,1\} \)
be any Boolean function and let $b=b(k)$ satisfy $b=\Omega(\log k)$, with a sufficiently large multiplicative constant. We define
$g:\{0,1\}^{b}\times\{0,1\}^{b}\to\{0,1\}$ by
\( g(u,v)=\langle u,v\rangle \bmod 2 \),
i.e., the inner product of two bit vectors taken modulo $2$. 

Define $F_k=f\circ g^{\,k}$ on inputs $(x,y)\in(\{0,1\}^{b})^{k}\times(\{0,1\}^{b})^{k}$ by
\begin{equation}
\label{eq:def-F}
F_k(x,y)= f\bigl(g(x^{(1)},y^{(1)}),\ldots,g(x^{(k)},y^{(k)})\bigr),
\end{equation}
where each $x^{(j)}$ and $y^{(j)}$ is an array of $b$ input variables. Then
\begin{equation}
\ucc(F)\le \bigO \!\bigl(\ucert(f)\cdot b\bigr),
\label{eq:cc-upper-bound}
\end{equation}
and
\begin{equation}
\concc(F)\ge \Omega\!\bigl(\concert(f)\cdot b\bigr).
\label{eq:cc-lower-bound}
\end{equation}
\end{theorem}

The upper bound \eqref{eq:cc-upper-bound} comes from an explicit protocol: given a fixed unambiguous DNF~$\varphi$ for $f$ of minimal width, the referee chooses one clause from $\varphi$ (which is supposed to certify that the value of $f$ is~\emph{True}), and broadcasts the chosen clause together with a valuation of its variables and the corresponding pairs $(x^{(j)}, y^{(j)})$; Alice and Bob then confirm that the broadcasted clause is indeed a clause in $\varphi$, that the broadcasted parts of their input are correct, and that the broadcasted valuation corresponds to the inner products $g(x^{(j)}, y^{(j)})$.

The lower bound \eqref{eq:cc-lower-bound}
is a much more involved result; it follows from~\cite[Theorem~2]{goos2016rectangles}.


Combined with Theorem~\ref{th:goos}, this theorem imples:
\begin{corollary}[Communication separation for lifted functions]
  \label{cor:goos}
Let ${\cal K}\subset \mathbb{N}$ be a set of natural numbers, and let $\{f_k\}_{k\in{\cal K}}$ be a family of Boolean functions
\( f_k : \{0,1\}^k \to \{0,1\} \)
such that for some $\alpha<\beta$, we have
\( \ucert(f_k)\;\le\; \bigO (k^\alpha) \)
and
\( \concert(f_k)\;\ge\;\Omega(k^\beta) \).

Let $b=b(k)$ satisfy $b=\Omega(\log k)$, and let $g:\{0,1\}^{b}\times\{0,1\}^{b}\to\{0,1\}$ be the inner-product gadget $g(u,v)=\langle u,v\rangle\bmod 2$.
Define
\( F_k := f_k\circ g^{\,k} \)
as in~Theorem~\ref{th:bounds_for_cc}, with one-party input length
\( n: = k\cdot b(k) \); then
\(
 \ucc(F_k)\;\le\; \bigO\!\bigl(k^\alpha\cdot b(k) \bigr) 
\)
and
\(
\concc(F_k)\;\ge\;\Omega\!\bigl(k^{\beta}\cdot b(k) \bigr).
\)
\end{corollary}

\begin{remark}
In the well-known work \cite{goos2015lower}, a separation between \( \ucc \) and \( \concc \) was established for the very natural problem {\sc Clique vs. Independent Set}.
Unfortunately, this problem has logarithmic \( \ucc \). This is too low for our applications, so we cannot base our proof on this result.
\end{remark}

\begin{remark}
\label{rem:log-gadgets}
\label{rem:log-gadgets-2}
The smaller the parameter $b=b(k)$, the bigger the relation between the lower and the upper bound in Corollary~\ref{cor:goos}.
 In the literature (see {\cite{goos2015lower,goos2016rectangles,goos2022lower}}),
the authors employ this machinery with the smallest possible $b(k)$ (that is $b(k) = \lambda \log k$ for some large enough constant $\lambda$). We, on the contrary, will rely on a very large $b(k)$ to prove Theorem~\ref{th:main}.
\end{remark}

\section{Proof of the main result}
\label{sec:main-proof}

\subsection{Construction}
\label{sec:construction}

In order to prove \Cref{th:main}, we will create a picture language from the communication problem given by \Cref{cor:goos}. First of all, we define a simple scheme that assigns picture languages to families of Boolean functions.

\begin{definition}
\label{def:Lfh}
Let  \( n \colon \mathbb{N} \to \mathbb{N} \) and  \( h \colon \mathbb{N} \to \mathbb{N} \) be two functions of a natural argument.
For a family \( \mathbf{F} = \{F_k\}_{k \in \mathcal{K}} \) of functions \( {F_k \colon \{0,1\}^{n(k)} \times \{0,1\}^{n(k)} \to \{0,1\}} \)
we define the picture language \(L_{\mathbf{F},h}\) over the alphabet ${\alphabet = \{0,1,\Box\}}$ composed of patterns $w$ on $2n(k)$ columns and $h(k)$ rows such that (see~\Cref{fig:generic-pattern}):
\begin{itemize}
\item the bottom row of $w$ consists of $2\cdot n(k)$ cells that can be either $0$ or $1$;  these cells represent two bit strings $x_1\ldots x_{n(k)}$ and $y_1\ldots y_{n(k)}$ such that
$F_k(x_1\ldots, x_{n(k)}, y_1\ldots y_{n(k)}) =1$;
we assume that the bits $y_1\ldots y_{n(k)}$ are presented in the picture in reverse order, see Fig.~\ref{fig:generic-pattern}.

\item all cells of the other $h-1$ rows of the patterns are blank symbols $\Box$.
\end{itemize}
\end{definition}

\noindent If we consider the communication problem in which Alice and Bob are each given a pattern of size $n(k) \times h(k)$, and want to check if the concatenation of these patterns forms a picture in \(L_{\mathbf{F},h}\), we essentially re-obtain the communication problem given by the family of functions \( \mathbf{F} \).
Indeed, it reduces to giving Alice and Bob the bottom rows of their patterns, and computing the value of the corresponding Boolean function \(F_k\) on these inputs.

\begin{figure}[htb]
  \begin{center}

\begin{tikzpicture}[
    scale=1.00,
    transform shape,
    line cap=round,
    line join=round,
    font=\small
  ]

  \def\N{32}       
  \pgfmathtruncatemacro{\NplusOne}{\N+1}
  \pgfmathtruncatemacro{\TwoN}{2*\N}
  \def\H{4.0}      
  \def\sq{0.13} 
  \def\s{0.20}     
  \def\eps{0.04}     

 \pgfmathsetmacro{\W}{2*\N*\s}

  \fill[gray!15]
  (0,0) rectangle (\W,\H);

  \draw[<->]
    (0+\eps,-0.45) -- ({\W/2-\eps},-0.45)
    node[midway,fill=white,inner sep=1pt] {\(n\)};

  \draw[<->]
    ({\W/2+\eps},-0.45) -- (\W-\eps,-0.45)
    node[midway,fill=white,inner sep=1pt] {\(n\)};

  \draw[
    decorate,
    decoration={brace, amplitude=3pt}
  ]
    (0.1,0.35) -- ({\W/2-0.1},0.35)
    node[midway, yshift=10pt]
    {bits \(x_1x_2 \ldots x_{n-1} x_n\)};

  \draw[
    decorate,
    decoration={brace, amplitude=3pt}
  ]
    ({\W/2+0.1},0.35) -- ({\W-0.1},0.35)
    node[midway, yshift=10pt]
    {bits \(y_ny_{n-1}\ldots y_2 y_1\)};

  \draw[<->]
    (\W+0.45,\eps) -- (\W+0.45,\H-\eps)
    node[midway,fill=white,inner sep=1pt] {\(h\)};

  \draw[thin] (0,0) -- (0,-0.6);
  \draw[thin] ({\W/2},0) -- ({\W/2},-0.6);
  \draw[thin] (\W,0) -- (\W,-0.6);

  \draw[thin] (\W,0) -- (\W+0.6,0);
  \draw[thin] (\W,\H) -- (\W+0.6,\H);

  \draw[black, ultra thick]
    (0,0) rectangle (\W,\H);

  \draw[black, thick, dashed]
    ({\W/2},0) -- ({\W/2},\H);

  \foreach \i in {1,...,\N} {
  \fill[green!60!black]
    ({(\i-1)*\s+0.1-\sq/2}, {0.1-\sq/2})
    rectangle
    ({(\i-1)*\s+0.1+\sq/2}, {0.1+\sq/2});
  }

  \foreach \i in {\NplusOne,...,\TwoN} {
    \fill[green!60!black]
    ({(\i-1)*\s+0.1-\sq/2}, {0.1-\sq/2})
    rectangle
    ({(\i-1)*\s+0.1+\sq/2}, {0.1+\sq/2});
  }

\end{tikzpicture}
  \end{center}
  \caption{Blueprint of a pattern of size \( 2n \times h \) for \(n := n(k)\) and \(h := h(k)\). The two blocks of green dots represent Alice's bits \(\mathbf{x}=x_1\ldots x_n\) and Bob's \(\mathbf{y}=y_1\ldots y_n\).
  The bits of Bob's input are presented in reverse order.
  The pattern belongs to the picture language \( L_{\mathbf{F},h} \) iff \( F_k(\mathbf{x},\mathbf{y})=1 \).}
\label{fig:generic-pattern}
\end{figure}

\paragraph*{Fixing principal parameters of the construction}

In this section, we fix the family \( \mathbf{F} = \{F_k\}_{k \in \mathcal{K}} \) and the function \( h \colon \mathbb{N} \to \mathbb{N} \) as follows:

\begin{itemize}\label{fixed-parameters}
\item We set an explicit family \( \mathbf{f} = \{ f_k \}_{k \in \mathcal{K}} \) given by~\Cref{th:goos}, \textit{i.e.}~whose complexity satisfies \( \ucert(f_k) = \bigO(k^{\alpha}) \) and \( \concert(f_k) = \Omega(k^{\beta}) \) for some \(\alpha < \beta\).
\item We fix \(b(k)\) such that \(b(k) = 2^{2^{\Theta(k)}}\) and  \(k \mapsto k \cdot b(k)\) is injective, and define \(n(k) := k \cdot b(k)\).
\item We define the family \( \mathbf{F} = \{F_k\}_{k \in \mathcal{K}} \) as \( F_k = f_k \circ g^{k} \), where \(g\) the inner-product gadget on bit-strings of length \(b(k)\).
\item Finally, we fix a function \(h(k): =  \lfloor k^\gamma \cdot b(k) \rfloor \)
such that
\( \ucc(F_k)  \ll h(k) \ll \concc(F_k) . \)
\end{itemize}
Without loss of generality we will assume that $b(k)$ and $h(k)$ are fully time-constructible functions\footnote{%
A function $f\colon\mathbb{N}\to\mathbb{N}$ is called fully time-constructible if there exists a Turing machine that on every input of length $n$
runs for exactly $f(n)$ steps, see, \textit{e.g.}, \cite{homeri2001computability}.
If a function $f\colon\mathbb{N}\to\mathbb{N}$ is computable in time $\bigO(f(n))$ and satisfies $f(n)\geq (1+\varepsilon)n$ for some $\varepsilon>0$ and all sufficiently large $n$,
then it is fully time-constructible~\cite{kobayashi1985proving}. In particular, most commonly used superlinear functions,
such as $f(n)=n^k$ for $k=2,3,\ldots$, $f(n)=2^n$, or $f(n)=2^{2^n}$, are fully time-constructible.%
}.

By interpreting the patterns of the resulting language \(L_{\mathbf{F},h}\) as instances of a communication problem \(F_k \colon \{0,1\}^{n(k)} \times \{0,1\}^{n(k)} \to \{0,1\}\), we obtain the following results:
\begin{itemize}
\item In Section~\ref{sec:proof-negative-claim} we prove that \(L_{\mathbf{F},h}\) is not in \(\corec\). This is the main contribution of the paper.
  The proof is based upon the fact that \( \ncc(\neg F_k) \gg  h(k)\), so that the frontier between Alice's and Bob's halves of the patterns (the dashed line in~\Cref{fig:generic-pattern}) is too small to encode
communication transcripts.
\item In \Cref{sec:proof-positive-claim-short} we provide a proof that \(L_{\mathbf{F},h}\) is in \(\urec\). The argument is based on the fact that \(\ucc(F_k) \leq h(k)\): we construct a tiling system that ``geometrically'' simulates a specific communication protocol for \(\{ F_k \}_{k \in \mathcal{K}}\) between Alice and Bob.
\end{itemize}

\subsection{\boldmath Proof that \(L_{\mathbf{F},h}\) cannot be in \(\corec\)}
\label{sec:proof-negative-claim}

\begin{lemma}
\label{lem:lower-bound}
Let $\mathbf{F} = \{ F_k \}_{k \in \mathcal{K}}$ be a family of Boolean functions
\(
F_k : \{0,1\}^{n(k)} \times \{0,1\}^{n(k)} \to \{0,1\},
\)
and
\(
h: \mathbb{N} \to \mathbb{N}
\)
be a function on natural numbers. Let $L_{\mathbf{F},h}$ be the picture language defined for these $F_k$ and $h$, as explained in Section~\ref{sec:construction}. If $L_{\mathbf{F},h}$ is in $\rec$, then the nondeterministic communication complexity of $F_k$ is at most $\bigO(h(k))$.
\end{lemma}

\begin{remark}
Such a condition was first observed in~\cite{anselmo2013stronger}. 
It was also an important ingredient of the argument in \cite{terrier2019communication}.
\end{remark}

\begin{proof}Assume that a picture language $L_{\mathbf{F},h}$ admits a tiling system (in the sense of Definition~\ref{def:rec}): there exists a local picture language $L$ over some finite alphabet $\Sigma$ and a mapping 
\( {\pi : \Sigma \to \{0,1,\Box\}} \)
such that any rectangular pattern $P$ belongs to $L_{\mathbf{F},h}$ if and only if there exists a pattern $Q \in L$ such that $P$ is a letter-wise $\pi$-projection of $Q$.

\medskip
For $(x,y)\in \{ 0,1\}^n \times \{0,1\}^n$ a pair of strings, we denote by $P_{x,y}$ the pattern of size $(2n) \times h(n)$ as shown in Figure~\ref{fig:generic-pattern} (the colors in the bottom row represent the bits of $x$ and $y$, and all other cells are blank). We define a communication protocol as follows:
\begin{itemize}
\item The referee chooses a picture $Q$ in $L$ such that $\pi(Q) = P_{x,y}$; and broadcasts the subpicture $Z$ of size $2\times (h+2)$ obtained as the two middle columns of the bordered picture $\hat Q$;
\item Alice accepts the referee's message if there exists a (partially bordered) picture $\hat{Q}_L$ of size $(n+2) \times (h+2)$ such that:
\begin{itemize}
\item $\hat{Q}_L$ is compatible with $Z$, which means that the two \emph{rightmost} columns of $\hat{Q}_L$ coincide with $Z$, as shown in Figure~\ref{fig:from-tiling-to-cc};
\item The picture $\hat{Q}_L$ satisfies the local constraints associated with the local language $L$;
\item The $\pi$-projection of the non-bordered $Q_L$ corresponds with Alice's input $x$;
\end{itemize}
\item Similarly, Bob accepts the referee's message if there exists a (partially bordered) picture~$\hat{Q}_R$ of size $(n+2) \times (h+2)$ such that:
\begin{itemize}
\item $\hat{Q}_R$ is compatible with $Z$, which means that the two \emph{leftmost} columns of $\hat{Q}_R$ coincide with $Z$, as shown in Figure~\ref{fig:from-tiling-to-cc};
\item The picture $\hat{Q}_R$ satisfies the local constraints associated with the local language $L$;
\item The $\pi$-projection of the non-bordered $Q_R$ corresponds with Bob's input $y$ (in reverse);
\end{itemize}
\end{itemize}
Observe that the patterns $\hat{Q}_L$ and $\hat{Q}_R$ compatible with the broadcasted $Z$ may be not unique, and may not coincide with the bordered picture $\hat{Q}$ chosen by the referee.

By construction, the communication complexity of the protocol is $\bigO(h)$, since the referee's message consists of $2(h+2)$ letters from a finite alphabet $\Sigma$. We now prove that this defines a valid protocol for $F_k$.

\smallskip
If $F_k(x,y) = 1$, then this protocol accepts the pair $(x,y)$. Indeed $P_{x,y} \in L_{\mathbf{F},h}$, so there exists a valid pattern $Q \in L$ such that $P_{x,y} = \pi(Q)$. The referee can broadcast the two central columns of $\hat{Q}$; and Alice and Bob can use the left and right halves of $\hat{Q}$ as $\hat{Q}_L$ and~$\hat{Q}_R$, which are locally valid and coincide with their respective inputs $x$ and~$y$ (see Figure~\ref{fig:from-tiling-to-cc}).

Conversely, we prove that if a pair $(x,y)$ is accepted by the protocol, then $F_k(x,y) = 1$. Indeed, if a broadcasted pattern $Z$ is accepted by Alice and Bob, then this $Z$ can be extended into a locally valid bordered picture $\hat{Q}$ such that $\pi(Q) = P_{x,y}$: indeed, there then there exist two $(n+2) \times h$ partially bordered patterns $\hat{Q}_L$ and $\hat{Q}_R$ extending $Z$ to the left and the right respectively; furthermore, their combination $\hat{Q}$ into a $(2n+2)\times (h+2)$ bordered picture is locally valid; so that the pattern $Q$ actually satisfies $Q \in L$, and then $P_{x,y} = \pi(Q) \in L_{\mathbf{F},h}$. By definition of $L_{\mathbf{F},h}$, we conclude that $F_k(x,y)=1$.
\end{proof}

\vspace{-0.6em}

\begin{figure}[bth]
\begin{center}
\begin{tikzpicture}[scale=0.38]
  \def\n{15}
  \def\h{7}

  \fill[blue!20] (0,0) rectangle (\n,\h);
  \fill[orange!25] (\n,0) rectangle (2*\n,\h);

  \filldraw[
    fill=green!70,
    fill opacity=0.3,
    draw=green!70!black,
    line width=1pt,
    rounded corners=0.25
  ]
    (\n-1-0.15,-0.15)
    rectangle
    (\n+1+0.15,\h+0.15);

  \draw[step=1, gray!60, thin] (0,0) grid (2*\n,\h);

  \draw[black, thick] (0,0) rectangle (2*\n,\h);



\draw[<->] (1.1,-1) -- (\n-0.1,-1)
  node[midway,fill=white,inner sep=1pt] {$n$};

\draw[thin] (\n,0) -- (\n,-1.2);

\draw[<->] (\n+0.1,-1) -- (2*\n-1.1,-1)
  node[midway,fill=white,inner sep=1pt] {$n$};


\draw[<->] (2*\n+1,1.1) -- (2*\n+1,\h-1.1)
  node[midway,fill=white,inner sep=1pt] {$h$};


\node[color=darkblue] at (0.25*\n,\h+1.8) {Alice's $Q_L$};
\draw[->, thick]
  (0.25*\n,\h+1)
  -- (0.35*\n,0.6*\h);

\node[color=darkorange] at (1.75*\n,\h+2) {Bob's $Q_R$};
\draw[->, thick]
  (1.75*\n,\h+1)
  -- (1.65*\n,0.6*\h);

\node[color=darkgreen] at (\n,\h+3) {referee's $Z$};
\draw[->, thick]
  (\n,\h+2)
  -- (\n-0.5,\h+0.15);

  \foreach \i in {1,...,\n} {
  \node at (\i -0.5,\h-0.5) { \tiny $\#$};
  \node at (\i -0.5+\n,\h-0.5) { \tiny $\#$};
  }

  \foreach \i in {1,...,\n} {
  \node at (\i -0.5,0.5) { \tiny $\#$};
  \node at (\i -0.5+\n,0.5) { \tiny $\#$};
  }

  \foreach \i in {2,...,\h} {
  \node at (0.5,\i-0.5) { \tiny $\#$};
  }

  \foreach \i in {2,...,\h} {
  \node at (\n+\n-0.5,\i-0.5) { \tiny $\#$};
  }

\end{tikzpicture}
\caption{The patterns from the proof of Lemma~\ref{lem:lower-bound}.}
\label{fig:from-tiling-to-cc}
\end{center}
\end{figure}

\vspace{-0.7em}

Lemma~\ref{lem:lower-bound} implies immediately the following corollary, which means that \(L_{\mathbf{F},h}\) constructed in Section~\ref{sec:construction}, is not in \(\corec\).
This concludes the proof of our main result.
\begin{corollary}
Let $\mathbf{F} = \{ F_k \}_{k \in \mathcal{K}}$ be a family of Boolean functions
\[
F_k : \{0,1\}^{n(k)} \times \{0,1\}^{n(k)} \to \{0,1\},
\]
and \( h: \mathbb{N} \to \mathbb{N} \)
be a function on natural numbers, and
let \(L_{\mathbf{F},h}\) be the picture language from Definition~\ref{def:Lfh}.
If $\concc(F_k) \gg h(k)$, then the picture language \(L_{\mathbf{F},h}\)  is not in $\corec$.
\end{corollary}

\subsection{\boldmath Proof that \(L_{\mathbf{F},h}\) is in \(\urec\)}
\label{sec:proof-positive-claim-short}

In this section  we design a local picture language $L$ over some alphabet $\alphabet$ and define a map
$
\pi:\alphabet\to\{0,1,\Box\}
$
such that the $\pi$-projections of pictures in $L$ are exactly the pictures in $L_{\mathbf{F},h}$, and every picture of $L_{\mathbf{F},h}$ has a unique preimage in $L$. This will prove that
$ L_{\mathbf{F},h}\in\urec$.

While the construction is a little technical, its individual ingredients are elementary and this proof should be routinely reproducible by researchers in multidimensional symbolic dynamics. We therefore only give a concise and schematic description of the construction.

At a high level, the construction simulates in $L$ the communication protocol for the proof of \( \ucc(F)\le \bigO \!\bigl(\ucert(f)\cdot b\bigr) \) in Theorem~\ref{th:bounds_for_cc}, which certifies the true values of the Boolean functions $F_k = f_k \circ g^k$. As is often the case in the construction of complex picture languages, we will describe the language $L$ as a superimposition of \emph{layers}. That is, the alphabet of $L$ will be a Cartesian product
$ \alphabet=\alphabet_0\times\alphabet_1\times\cdots\times\alphabet_{7}$,
and we will refer to the projection onto the component $\alphabet_i$ of a picture $Q \in L$ as the $i$\textsuperscript{th} layer of $Q$. 
This conceptual separation of pictures into layers allows for  a more intuitive description of the language $L$, as most local constraints in $L$ (though not all) enforce a prescribed structure within a single layer.

\smallskip
\emph{Layer 0: Base Layer.}
We set
$\alphabet_0=\{0,1,\Box\}$
and denote $\pi:\alphabet \to \{0,1,\Box\}$ the restriction onto the $\alphabet_0$-component.
This layer contains the picture from $L_{\mathbf{F},h}$ that we wish to certify.

\begin{figure}[t]
  \centering

  \begin{tikzpicture}[x=0.28cm,y=0.28cm,scale=0.68]

    \def\nn{32}
    \def\hh{24}
    \def\bb{8}
    \pgfmathtruncatemacro{\WW}{2*\nn}

    \foreach \x in {0,...,63} {
      \foreach \y in {0,...,23} {
        \fill[black!5]
        (\x,\y) rectangle ++(1,1);
      }
    }
    \foreach \y in {0,8,16} {
      \fill[horcolor]
      (0,\y) rectangle (\WW,\y+1);
    }
    \foreach \x in {0,8,16,24,32,40,48,56} {
      \fill[vertcolor]
      (\x,0) rectangle (\x+1,\hh);
    }

    \foreach \x in {0,...,63} {
      \foreach \y in {0,...,23} {
        \pgfmathtruncatemacro{\rr}{mod(\x-\y+16,8)} 
        \ifnum\rr=0
          \fill[diagcolor]
          (\x,\y) rectangle ++(1,1);
        \fi
      }
    }

    \foreach \x in {0,8,16,24,32,40,48,56} {
      \foreach \y in {0,8,16} {
        \fill[nodescolor]
        (\x,\y) rectangle ++(1,1);
        \fill[black]
        (\x+0.5,\y+0.5) circle[radius=0.15];
      }
    }

    \draw[step=1,thin,black!35, opacity=.7]
    (0,0) grid (\WW,\hh);

    \draw[thick]
    (0,0) rectangle (\WW,\hh);

    \draw[{Latex[length=3pt]}-{Latex[length=3pt]}] (\WW.8,0) -- ++(0,\bb) node[midway,right,font=\scriptsize] {$b$};
  \end{tikzpicture}
  \caption{Schematics of a valid picture in $L$: Layer 1 draws a square grid.}
  \vspace{-0.8em}
  \details{In Layer 1, vertical \textcolor{vertcolor}{$\blacksquare$} lines and horizontal \textcolor{horcolor}{$\blacksquare$} form a grid. Diagonal \textcolor{diagcolor}{$\blacksquare$} lines running from southwest to northeast enforce that each cell of this grid is a square, and that all squares are of the same size, as we enforce with the local rules that whenever two of these three types of lines (horizontal, vertical, diagonal) meet, the third type must pass through the same point~\crossbox. We also enforce in the local rules that the boundary of the picture aligns with the grid.}
  \label{fig:grid}

\vspace{1em}

  \begin{tikzpicture}[x=0.28cm,y=0.28cm,scale=0.68]

    \def\nn{32}
    \def\hh{24}
    \def\bb{8}
    \pgfmathtruncatemacro{\WW}{2*\nn}

    \foreach \x in {0,...,63} {
      \foreach \y in {0,...,23} {
        \fill[black!5]
        (\x,\y) rectangle ++(1,1);
      }
    }
    \foreach \y in {0,8,16} {
      \fill[horcolor!43]
      (0,\y) rectangle (\WW,\y+1);
    }
    \foreach \x in {0,8,16,24,32,40,48,56} {
      \fill[vertcolor!43]
      (\x,0) rectangle (\x+1,\hh);
    }

    \foreach \x in {0,...,63} {
      \foreach \y in {0,...,23} {
        \pgfmathtruncatemacro{\rr}{mod(\x-\y+16,8)} 
        \ifnum\rr=0
          \fill[diagcolor!43]
          (\x,\y) rectangle ++(1,1);
        \fi
      }
    }

    \foreach \x in {0,8,16,24,32,40,48,56} {
      \foreach \y in {0,8,16} {
        \fill[nodescolor!43]
        (\x,\y) rectangle ++(1,1);
        \fill[black!43]
        (\x+0.5,\y+0.5) circle[radius=0.15];
      }
    }

    \draw[step=1,thin,black!35, opacity=.7]
    (0,0) grid (\WW,\hh);

    \foreach \x in {31} {
      \foreach \y in {1,...,\hh} {
        \path[
        draw=midcolor,
        line width=0.7pt, 
        pattern={Lines[
          angle=45,
          distance=1.5pt,  
          line width=0.3pt 
          ]},
        pattern color=midcolor
        ] (\x,\y-1) rectangle (\x+1,\y);
      }}

    \foreach \x in {32} {
      \foreach \y in {1,...,\hh} {
        \path[
        draw=midcolor,
        line width=0.7pt, 
        pattern={Lines[
          angle=135,
          distance=1.5pt,  
          line width=0.30pt 
          ]},
        pattern color=midcolor
        ] (\x,\y-1) rectangle (\x+1,\y);
      }}

    \foreach \x in {0,...,30} {
      \foreach \y in {0,...,23} {
        \draw[-{Stealth[length=1mm,width=1mm]}, midcolor!80]
       (\x+.5,\y+.5) -- ++(0.3,0); 
      }
    }

    \foreach \x in {33,...,63} {
      \foreach \y in {0,...,23} {
        \draw[-{Stealth[length=1mm,width=1mm]}, midcolor!80]
       (\x+.5,\y+.5) -- ++(-0.3,0); 
      }
    }

    \draw[thick]
    (0,0) rectangle (\WW,\hh);

    \draw[{Latex[length=3pt]}-{Latex[length=3pt]}] (\WW.8,0) -- ++(0,\bb) node[midway,right,font=\scriptsize] {$b$};
  \end{tikzpicture}
  \caption{Schematics of a valid picture in $L$:  Layer 2 selects in the grid a pair of adjacent vertical strips.}
    \vspace{-0.8em}
  \details{In Layer~2, we use four tile types: $\larrowsquare$ (filling the leftmost columns of the picture), \brownboxslash\ (filling the next column), \brownboxbackslash\ (filling the column immediately after it), and $\rarrowsquare$ (filling all remaining columns to the right boundary of the picture).
The two adjacent columns of \brownboxslash\ and~\brownboxbackslash\ symbols form a pair of vertical strips called the \emph{Middle Line}. The local rules require the column of \brownboxbackslash\ symbols to coincide with a vertical \textcolor{vertcolor}{$\blacksquare$} column from Layer~1.
The local rules for Layer~2 do not uniquely determine the position of the Middle Line: additional constraints will be introduced to ensure that it lies at the horizontal center of the picture, with the same number of columns on either side.}
  \label{fig:mid}

\end{figure}

\smallskip
\emph{Layer 1: Grid Layer.}
We do not specify the alphabet $\alphabet_1$ explicitly, but instead explain the geometric patterns implemented on this layer. Layer 1 draws a regular grid of square-shaped cells with vertical \textcolor{vertcolor}{$\blacksquare$} lines and horizontal \textcolor{horcolor}{$\blacksquare$}, and a system of diagonal \textcolor{diagcolor}{$\blacksquare$} lines enforce the regularity and the squared shape of the grid cells. See Figure~\ref{fig:grid} for more details.

In the rest of this proof, and for any valid picture $Q \in L$, we denote $b$ the grid spacing of the 1\textsuperscript{st} Layer of $Q$. 
At this stage, the spacing $b$ is not fixed: its value will be determined below through constraints coupling Layer~1 with the other layers.

\smallskip
\emph{Layer 2: Middle Line.}
On this layer, we mark two adjacent columns with $\brownboxslash$ and $\brownboxbackslash$ tiles (shown in brown in Figure~\ref{fig:mid}) and refer to them collectively as \emph{the Middle Line}. While one can enforce with local rules  that there is exactly one Middle Line, its position is not fixed at this stage: the middle position will be enforced by Layer 3.

\newpage
\enlargethispage{\baselineskip}
\smallskip
\emph{Layer 3: Block Counting.}
The grid from Layer~1 (superimposed on the bottom row of Layer~0) partitions the bottom row of Layer 0 into consecutive \emph{blocks of $b$ bits}.
We denote these blocks by $x^{(1)},x^{(2)},\ldots$ (counting form left to right) if they appear on the left of the Middle Line, and by $y^{(1)},y^{(2)},\ldots$  (counting from right to left) if they appear on its right.

In \emph{Layer 3}, we introduce a geometric structure assigning an \emph{index} to each such block. The indices are represented by simple unary counters (using $\unarymarkerbox$ symbols) and enumerate the blocks from left to right, starting at $1$. Thus, every $b$-bit block in the bottom row of Layer~0 is assigned an index. See Figure~\ref{fig:block_indices} for more details.

This construction naturally provides us with another useful feature: we can ensure that the two halves of a valid pattern have the same number of  $b$-bit blocks, and hence that the Middle Line indeed bisects the picture. In the rest of this proof, we denote $k'$ the maximal value of the counter, which corresponds to the height of the bold blue line in Figure~\ref{fig:block_indices}.

\begin{figure}[t]
  \centering
  \begin{tikzpicture}[scale=.21]
    \begin{scope}

      \foreach \x in {0,8,16,24,32,40,48,56} {
        \fill[vertcolor!70]
        (\x,0) rectangle (\x+1,24);
      }

      \draw[TilingGrid,opacity=.99] (0,0) grid ++(64,24);

    \end{scope}
    %
    \draw[{Latex[length=3pt]}-{Latex[length=3pt]},font=\small] (64.5,0) -- ++(0,4) node[midway,right] {$k'$};
    \foreach \k in {0,1,2,3} {
      \pgfmathtruncatemacro{\kn}{\k+1}
      \draw[{Latex[length=3pt]}-{Latex[length=3pt]},font=\small] ($(8*\k-.1,-.5)$) -- ++($(\k+1.2,0)$) node[midway,below] {$\kn$};
      \draw[{Latex[length=3pt]}-{Latex[length=3pt]},font=\small] ($(63 - 8*\k + 1.2,-.5)$) -- ++($(- \k-1.2,0)$) node[midway,below] {$\kn$};
    }

    \foreach \k in {0,1,2,3} {
      \pgfmathsetmacro{\i}{9*\k}
      \draw[unarycountercolor,thick] (\i,0) rectangle ++(1,1);
      \draw[unarycountercolor,thick] (\i,0) -- ++(1,1) (\i,1) -- ++(1,-1);
      \draw[unarycountercolor,thick] ($(63-\i,0)$) rectangle ++(1,1);
      \draw[unarycountercolor,thick] ($(63-\i,0)$) -- ++(1,1) ($(63-\i,1)$) -- ++(1,-1);
    }

    \draw[very thick,countbcolor!70] (0,3.75) -- ++(64,0);
    \begin{scope}[thick,densely dashed,countkcolor]
      \coordinate (p) at (0.5,0.5);
      \foreach \k in {1,2,3} {
        \draw[-{Stealth[length=2.5mm,width=1.2mm]}]  (p) -- ++(8,0) coordinate (p);
        \draw     (p) -- ++(0,1) coordinate (p);
      }
      \draw[-{Stealth[length=2.5mm,width=1.2mm]}]  (p) -- ++(7,0);
    \end{scope}

    \foreach \k in {1,2,3} {
      \begin{scope}[
        thick,densely dashed,countkcolor,
        >={Stealth[length=2mm,width=1.2mm]}
        ]
        \draw[thick,densely dashed,countkcolor2,->] ($(0.5,0.5) + (8*\k,\k)$) -- ++(\k-0.2,-\k+0.2);
        \draw[thick,densely dashed,countkcolor2,->] ($(63+0.5,0.5) + (-8*\k,\k)$) -- ++(-\k+0.2,-\k+0.2);
      \end{scope}
    }
    \begin{scope}[thick,densely dashed,countkcolor]
      \coordinate (p) at (63.5,0.5);
      \foreach \k in {1,2,3} {
        \draw[-{Stealth[length=2.5mm,width=1.2mm]}] (p) -- ++(-8,0) coordinate (p);
        \draw     (p) -- ++(0,1) coordinate (p);
      }
      \draw[-{Stealth[length=2.5mm,width=1.2mm]}]  (p) -- ++(-7,0);
    \end{scope}

    \draw[thick]
    (0,0) rectangle (64,24);
    
     \def\pas{0.89}
      \draw[<->] (0+0.1*\pas,-4.0) -- (36*\pas - 0.1*\pas,-4.0)
    node[midway,fill=white,inner sep=1pt] {$n$};
    
      \draw[<->] (36*\pas+0.05*\pas,-4.0) -- (72*\pas - 0.05*\pas,-4.0)
    node[midway,fill=white,inner sep=1pt] {$n$};

  \end{tikzpicture}
  \caption{Schematics of a valid picture in $L$: Layer 3 draws a system of $\unarymarkerbox$ unary markers, each marking the index of a block $x^{(i)}$ and $y^{(i)}$.}
  \details{In the left half of the picture, the dashed lines represent a signal propagating from left to right, moving one step upward each time it crosses a vertical $\textcolor{vertcolor}{\blacksquare}$ line,
    and then being reflected toward the bottom row to place the corresponding $\unarymarkerbox$ unary marker.
     These markers serve as indices of the $b$-bit blocks.
     In the right half of the picture the construction is symmetric.
   By requiring that the counters of the left and right halves have the same height when entering the Middle Line, we ensure that the Middle Line indeed bisects the picture. We draw a horizontal bold \begin{tikzpicture}[scale=.21]\clip (0,0) rectangle ++(1,1); \draw[very thick,countbcolor!70] (0,.5) -- ++(1,0);\end{tikzpicture} line at this maximal height $k'$. We will use this horizontal line once more when discussing Layer~7.}
  \label{fig:block_indices}
\end{figure}

\enlargethispage{3\baselineskip}
\smallskip
\emph{Layer 4: Block Selection.}
On this layer,  some of the $b$-bit blocks are marked as \emph{selected}.
The bits of each selected block, together with its index, are transmitted by a system of ``wires'' towards the vertical Middle Line.
The local rules ensure that blocks with the same indices have to be selected on the two sides from the Middle Line.
Thus, in a valid picture of the language $L$, there exists a sequence of indices
$i_1,\ldots,i_t$
such that the selected blocks are
$x^{(i_1)},\ldots,x^{(i_t)}$
on the left and
$y^{(i_1)},\ldots,y^{(i_t)}$
on the right.
For more details, see Figure~\ref{fig:bittransmission}.

\begin{figure}[h]
  \centering

  \begin{tikzpicture}[
    scale=0.75,
    transform shape,
    line cap=round,
    line join=round,
    font=\small
    ]
    \def\k{6}
    \def\b{6}
    \def\pas{.25}
    \pgfmathsetmacro{\db}{\b*\pas}
    \pgfmathsetmacro{\W}{\db*\k}
    \def\H{5.0}      

    \colorlet{col1}{blue!60!black}
    \colorlet{col4}{red!60!black}
    \colorlet{col5}{cyan!60!black} 

    \tikzset{
      on each segment/.style={
        decorate,
        decoration={
          show path construction,
          lineto code={
            \path [#1]
            (\tikzinputsegmentfirst) -- (\tikzinputsegmentlast);
          },
        },
      },
      marrow/.style={postaction={decorate,decoration={
            markings,
            mark=at position .5 with {\arrow[#1]{Latex[length=1mm, width=1mm]}}
          }}},
    } 

    \begin{scope}

      \foreach \x in {0,6,12,18,24,30,42,48,54,60,66} {
        \fill[vertcolor!70]
        (\x*\pas,0) rectangle (\x*\pas+\pas,20 * \pas);
      }
    \end{scope}

    \pgfmathtruncatemacro{\kp}{\k-1}
    \fill[green!20!lightgray] (0,0) rectangle (2*\W,\pas);

    \foreach \c [count=\i from 0] in {5,4,1} {
      \fill[{col\c}]  (\c*\db-\db,0) rectangle (\c*\db,\pas) ;
      \fill[{col\c}]  (\W-\pas,\i*\db) rectangle ++ (\pas,\db);
      \fill[{col\c}]  (\W,\i*\db) rectangle ++ (\pas,\db);
    }

    \draw[TilingGrid,xstep=\pas,ystep=\pas] (0,0) grid ($(2*\W,\H)$);

    \foreach \x in {35} {
      \foreach \y in {19,20} {
        \path[
          draw=midcolor,
          line width=0.7pt, 
          pattern={Lines[
            angle=45,
            distance=1.5pt,  
            line width=0.3pt 
            ]},
          pattern color=midcolor
        ] (\x*\pas,\y*\pas-\pas) rectangle (\x*\pas+\pas,\y*\pas);
      }}

    \foreach \x in {36} {
      \foreach \y in {19,20} {
        \path[
          draw=midcolor,
          line width=0.7pt, 
          pattern={Lines[
            angle=135,
            distance=1.5pt,  
            line width=0.3pt 
            ]},
          pattern color=midcolor
        ] (\x*\pas,\y*\pas-\pas) rectangle (\x*\pas+\pas,\y*\pas);
      }}

    \foreach \c/\e in {1/1,4/2,5/t} {
      \draw[
      decorate,
      decoration={brace, mirror, amplitude=3pt},
      col\c
      ]
      (\pas*.5+\c*\db-\db,-\pas) --++ (\db-\pas,0)
      node[midway, yshift=-10pt]
      {select\,\(x^{(i_\e)}\)};
    }

    \foreach \ind in {1,...,\b} {
      \draw[col1,postaction={on each segment={marrow}}]  (-\pas*.5+\ind*\pas,\pas*.5) -- (-\pas*.5+1*\db,\db-\ind*\pas+\pas*.5) --++ (\db,0) --++ (\db,0) --++ (\db,\db) --++ (\db,\db) --++ (\db,0);
      \draw[col4,postaction={on each segment={marrow}}]  (-\pas*.5+3*\db+\ind*\pas,\pas*.5) -- (-\pas*.5+4*\db,\db-\ind*\pas+\pas*.5)--++ (\db,\db) --++ (\db,0);
      \draw[col5,postaction={on each segment={marrow}}]  (-\pas*.5+4*\db+\ind*\pas,\pas*.5) -- (-\pas*.5+5*\db,\db-\ind*\pas+\pas*.5) --++ (\db,0);
    }

    \begin{scope}[shift={(\W,0)}]
      \foreach \c [count=\i from 0] in {5,4,1} {
        \fill[{col\c}]  (42*\pas-\c*\db-\db,0) rectangle (42*\pas - \c*\db,\pas) ;
      }

      \foreach \c/\e in  {1/1,4/2,5/t} {
        \draw[
        decorate,
        decoration={brace, mirror, amplitude=3pt},
        col\c
        ]
        (\pas*0.5- \c*\db+6*\db  ,-\pas) --++ (\db-\pas,0)
        node[midway, yshift=-10pt]
        {select\,\(y^{(i_\e)}\)};
      }

      \begin{scope}[shift={(6*\db,0)},xscale=-1]
        \foreach \ind in {1,...,\b} {
          \draw[col1,postaction={on each segment={marrow}}]  (-\pas*.5+\ind*\pas,\pas*.5) -- (-\pas*.5+1*\db,\db-\ind*\pas+\pas*.5) --++ (\db,0) --++ (\db,0) --++ (\db,\db) --++ (\db,\db) --++(\db,0);
          \draw[col4,postaction={on each segment={marrow}}]  (-\pas*.5+3*\db+\ind*\pas,\pas*.5) -- (-\pas*.5+4*\db,\db-\ind*\pas+\pas*.5)--++ (\db,\db)--++(\db,0);
          \draw[col5,postaction={on each segment={marrow}}]  (-\pas*.5+4*\db+\ind*\pas,\pas*.5) -- (-\pas*.5+5*\db,\db-\ind*\pas+\pas*.5)--++(\db,0);
        }
      \end{scope}
    \end{scope}

    \foreach \i in {1,...,\k}{
      \draw[unarycountercolor,thick] (\i*\db-\db+\i*\pas-\pas,0) rectangle +(\pas,\pas) (\i*\db-\db+\i*\pas-\pas,0) -- ++(\pas,\pas) ++(0,-\pas) --++(-\pas,\pas);
      \draw[unarycountercolor,thick] (\i*\db-\db+\i*\pas-\pas+\W,0) rectangle +(\pas,\pas) (\i*\db-\db+\i*\pas-\pas+\W,0) -- ++(\pas,\pas) ++(0,-\pas) --++(-\pas,\pas);
    }
    \foreach \i/\j in {1/1,4/2,5/3}{
      \draw[unarycountercolor,very thick] (\W,4*\db-\j*\db-\i*\pas) rectangle +(\pas,\pas) (\W,4*\db-\j*\db-\i*\pas) -- ++(\pas,\pas) ++(0,-\pas) --++(-\pas,\pas);
      \draw[unarycountercolor,very thick] (\W-\pas,4*\db-\j*\db-\i*\pas) rectangle +(\pas,\pas) (\W-\pas,4*\db-\j*\db-\i*\pas) -- ++(\pas,\pas) ++(0,-\pas) --++(-\pas,\pas);
    }

    \draw[thick]
    (0,0) rectangle (18,5);

        \draw[<->] (0+0.1*\pas,-1.2) -- (36*\pas - 0.1*\pas,-1.2)
    node[midway,fill=white,inner sep=1pt] {$n$};
    
      \draw[<->] (36*\pas+0.05*\pas,-1.2) -- (72*\pas - 0.05*\pas,-1.2)
    node[midway,fill=white,inner sep=1pt] {$n$};
    
  \end{tikzpicture}
  \caption{Schematics of a valid picture in $L$: Layer 4 selects blocks of bits $x^{(i_j)}$ and $y^{(i_j)}$ and wires them towards to the vertical Middle Line.}
  \details{On the left half of a picture, we implement ``wires'' heading either right or diagonally until they arrive to the target vertical line.
    The value of the selected input bits (from Layer~0) together with the unary counter markers \( \unarymarkerbox \) (from Layer~3) on the bottom row spread along the corresponding wires.
The local constraints ensure that there is no gaps between consecutive wires 
(if a tile does not belong to any wire, then neither does the tile immediately above it).

    The symmetric scheme is used to transmit the selected~$y^{(i_j)}$ on the right half of the picture.}
  \label{fig:bittransmission}
\end{figure}

\smallskip
\emph{Layer 5: Inner Products.}
This layer processes the bits of the selected blocks from Layer~4 appearing on the Middle Line: the local rules implement a running-sum computation of the inner products over $\mathbb F_2$ for each selected pair $(x^{(i_j)},y^{(i_j)})$. The corresponding inner-product values are ``recorded'' along the Middle Line, at the bottom of each block.

\newpage
\smallskip
\emph{Layer 6: Certificate Transmission.}
This layer again implements signal propagation along geometric structures that plays the role of ``communication wires'', as shown in Figure~\ref{fig:certificatetransmission}, to collect the following data in a compact form drawn at the bottom-left corner of the picture:
\begin{itemize}
\item[(i)] the indices of the bit blocks selected by Layer~4 as a string $\mathtt{select\_ind} \in \{\Box,\unarymarkerbox\}^{\ast}$.
\item[(ii)] the inner-product values computed in Layer~5 for the corresponding pairs of selected blocks, in the form of a string $\mathtt{inner\_prod} \in \{0,1\}^{t}$.
\end{itemize}

\begin{figure}[h]
  \centering
  \begin{tikzpicture}[
    scale=0.88,
    transform shape,
    line cap=round,
    line join=round,
    font=\small
    ]
    \def\k{6}
    \def\b{6}
    \def\pas{.25}
    \pgfmathsetmacro{\db}{\b*\pas}
    \pgfmathsetmacro{\W}{\db*\k}
    \def\H{5.0}      
    \pgfmathsetmacro{\maxcols}{\H/\pas/\b}
    \colorlet{col1}{blue!60!black}
    \colorlet{col4}{red!60!black}
    \colorlet{col5}{cyan!60!black} 

    \tikzset{
      on each segment/.style={
        decorate,
        decoration={
          show path construction,
          lineto code={
            \path [#1]
            (\tikzinputsegmentfirst) -- (\tikzinputsegmentlast);
          },
        },
      },
      marrow/.style={postaction={decorate,decoration={
            markings,
            mark=at position .5 with {\arrow[#1]{Latex[length=1mm, width=1mm]}}
          }}},
    } 

    \foreach \x in {0,6,12,18,24,30,36} {
      \fill[vertcolor!50]
      (\x*\pas,0) rectangle (\x*\pas+\pas,20*\pas);
    }
    \foreach \c in {0,...,\maxcols} {
      \draw[horcolor,very thick]  (0,\c*\db) --++ (\W+\pas,0);
    }

    \fill[innerproductcolor] (0,0) rectangle +(3*\pas,\pas);

    \draw[TilingGrid,opacity=.6,xstep=\pas,ystep=\pas] (0,0) grid (\W+\pas,\H);

    \foreach \x in {35} {
      \foreach \y in {1,...,20} {
        \path[
          draw=midcolor,
          line width=0.7pt, 
          pattern={Lines[
            angle=45,
            distance=1.5pt,  
            line width=0.3pt 
            ]},
          pattern color=midcolor
        ] (\x*\pas,\y*\pas-\pas) rectangle (\x*\pas+\pas,\y*\pas);
      }}

    \foreach \x in {36} {
      \foreach \y in {1,...,20} {
        \path[
          draw=midcolor,
          line width=0.7pt, 
          pattern={Lines[
            angle=135,
            distance=1.5pt,  
            line width=0.3pt 
            ]},
          pattern color=midcolor
        ] (\x*\pas,\y*\pas-\pas) rectangle (\x*\pas+\pas,\y*\pas);
      }}

    \foreach \i/\j in {1/1,4/2,5/3}{
      \fill[innerproductcolor] (\W-\pas,\j*\db-\db) rectangle +(\pas,\pas);
      \fill[innerproductcolor] (\W,\j*\db-\db) rectangle +(\pas,\pas);
      \draw[TilingGrid,opacity=.6] (\W,\j*\db-\db) -- ++(0,\pas);

      \draw[unarycountercolor!80!black, postaction={on each segment={marrow}}, thick, dotted]  (\W-\pas*.5-.2*\pas,4*\db-\j*\db-\i*\pas+\pas*.5+.2*\pas) --++ (-\W+\i*\pas,0) --++(0,-4*\db+\j*\db+\i*\pas);
    }

    \foreach \i/\j in {1/1,4/2,5/3}{
      \draw[unarycountercolor,very thick] (\W-\pas,4*\db-\j*\db-\i*\pas) rectangle +(\pas,\pas) (\W-\pas,4*\db-\j*\db-\i*\pas) -- ++(\pas,\pas) ++(0,-\pas) --++(-\pas,\pas);
      \draw[unarycountercolor,very thick] (\W,4*\db-\j*\db-\i*\pas) rectangle +(\pas,\pas) (\W,4*\db-\j*\db-\i*\pas) -- ++(\pas,\pas) ++(0,-\pas) --++(-\pas,\pas);
      \draw[unarycountercolor,very thick] (\i*\pas-\pas,0) rectangle +(\pas,\pas) (\i*\pas-\pas,0) -- ++(\pas,\pas) ++(0,-\pas) --++(-\pas,\pas);
    }

    \begin{scope}[shift={(0,-.075)}]
      \foreach \ind in {1,2}{
        \draw[innerproductcolor!75!black,-latex,rounded corners, postaction={on each segment={marrow}}]  (\W-\pas*.5,\ind*\db+\pas*.8) --++ (-\W+3*\pas-\ind*\pas,0) --++ (0,-\ind*\db);
      }
      \draw[innerproductcolor!75!black,-latex,rounded corners, postaction={on each segment={marrow}}]  (\W-\pas*.5,\pas*.8) --++ (-\W+3*\pas,0);
    \end{scope}

    \foreach \c in {1,...,\maxcols}{
      \draw[unarycountercolor,thick,dashed] (0,\c*\db) -- +(\db,-\db);
    }

    \draw[{Latex[length=3pt]}-{Latex[length=3pt]},font=\small] (38.5*\pas,0*\pas) -- ++(0,6*\pas) node[midway,right] {$b$};
    \draw[{Latex[length=3pt]}-{Latex[length=3pt]},font=\small] (38.5*\pas,6*\pas) -- ++(0,6*\pas) node[midway,right] {$b$};
    \draw[{Latex[length=3pt]}-{Latex[length=3pt]},font=\small] (38.5*\pas,12*\pas) -- ++(0,6*\pas) node[midway,right] {$b$};

    \draw[<->] (0+0.1*\pas,-0.5) -- (36*\pas - 0.1*\pas,-0.5)
    node[midway,fill=white,inner sep=1pt] {$n$};

    \draw[thick]
    (9.25,5) -- (0,5) -- (0,0) -- (9.25,0);
  \end{tikzpicture}
  \caption{Schematics of a valid picture in $L$: Layer 6 wires the indices of the selected blocks $(x^{(i_j)}, y^{(i_j)})$ and the values of the corresponding inner products (as they appear on Layer 5) to a compact form appearing on the bottom-left corner of the picture.}
  \details{The result bits of the inner products, highlighted as \(\textcolor{innerproductcolor}{\blacksquare}\) in the figure, are moved horizontally and then downwards. The unary markers \(\unarymarkerbox\) (corresponding to the selected indices) are moved horizontally (dotted lines) until they cross a dashed diagonal (in the leftmost column of $b\times b$ squares), and then downwards, thus forming a bitmap of size $k$ containing $\unarymarkerbox$ at the selected positions and blanks otherwise.}
  \label{fig:certificatetransmission}
\end{figure}

\newpage
\smallskip
\emph{Layer 7: Computation Layer.}
This layer represents a space-time diagram of a Turing machine.
This machine operates on finitely many \emph{tapes}; each tape has some moving \emph{heads}, which read and write symbols on the tapes' cells (we may allow several heads in one tape).
We specify below the input supplied to this machine and the computation it performs.

We embed the space-time diagram of a Turing machine. The \emph{space} is horizontal: every row of the picture represents a Turing machine configuration that includes the data written on the tapes, the positions of the heads, and the internal state of the machine. The \emph{time} is vertical and goes from the bottom to the top.
The coherence of such a diagram can be easily implemented as a local picture language, see \textit{e.g.}, \cite{levin1996fundamentals,van1997convenience,lewis1998elements}.
We assume that the computation starts with all heads on the first cell (the leftmost tile of the row).
In the local rules, we also enforce that valid pictures only embed successfully terminating computations.

We now specify the $9$~tapes of our multitape Turing machine and their roles:
\begin{itemize}
\item one read-only tape contains the unary representation of an integer denoted by $k$;
  its value is guessed nondeterministically in bottom row of the diagram and is kept intact;

\item one tape takes as an input the gap between the two leftmost vertical $\textcolor{vertcolor}{\blacksquare}$ lines of the grid of Layer~1 (which can be interpreted as the parameter $b$ in the unary expansion);

\item two other tapes take as their inputs the bit strings $\mathtt{select\_ind}$ and $\mathtt{select\_ind}$, which are directly copied from Layer~6 (where they already appear on the bottom-left corner);

\item one tape uses input $k$ to count from $1$ to $b(k)$
  (recall that $b(k)$ is a fully time-constructible function)
  and makes sure that $b = b(k)$ matches as specified in the definition of $L_{\mathbf{F},h}$;

\item another tape uses input $k$ to count from $1$ to $h(k)$
  (recall that $h(k)$ is also a fully time-constructible function)
  and makes sure that this run is over exactly at the top row for the picture (the last row of the space-time diagram);

\item one more tape is used to check that the number of $b$-bit blocks encoded in both halves of the picture
  (\textit{i.e.}~the number of blocks $x^{(i)}$ and $y^{(j)}$) is equal to $k$;
  to this end, we check that the height $k'$ of the bold \begin{tikzpicture}[scale=.21]\clip (0,0) rectangle ++(1,1); \draw[very thick,countbcolor!70] (0,.5) -- ++(1,0);\end{tikzpicture} line on Layer~3 is equal to $k$ (kept on the first tape);

\item the last tape uses the compact certificates $\mathtt{select\_ind}$ and $\mathtt{inner\_prod}$ to certify that the function $F_k=f\circ g^{\,k}$ from Theorem~\ref{th:bounds_for_cc} is equal to $1$ on the bit strings embedded into the bottom row of the picture.
\end{itemize}
Let us provide some more details on this last item:
\begin{enumerate}
\item Using the last working tape, the machine first checks that $k$ actually belongs to the set of integers $\mathcal{K}$ indexing the family $\mathbf{f} = \{f_k\}_{k \in \mathcal{K}}$ fixed in \Cref{sec:construction}.
\item The machine then enumerates the list of all DNF formulas on $k$ variables,
and selects from this enumeration the unambiguous DNF $\varphi$ of minimal cost representing the Boolean function $f_k$ (if there are several such DNF, it takes the first in lexicographic order).
This search is double-exponential in term of the parameter $k$ (there are $2^{2^{\bigO( k )}}$ such formulas) but small compared to $b(k)$.
\item Finally, the machine checks that the indices encoded in $\mathtt{select\_ind}$ correspond exactly to the list of variables in one of the clauses in $\varphi$,
and the partial valuation assigning the truth values from $\mathtt{inner\_prod}$ to these variables evaluations the clause to true.
\end{enumerate}
If any of these checks fails, the machine should stop and reject the input; otherwise, it accepts.
Since the latter computation only depends on $k$, we may assume that the computation implemented on the last tape halts in time $\ll h(k)$. Therefore, we have enough room to embed the entire space-time diagrams in the pictures of size $h \times (2n)$.

\bigskip
\emph{End of the proof.}
It remains to verify that the constructed layers together provide an \emph{unambiguous} local language for $L_{\mathbf{F},h}$.

\begin{proof}
  For every picture $P_{x,y} \in L_{\mathbf{F},h}$, we prove that there exists some picture $Q \in L$ such that $\pi(Q) = P$. Indeed, let $k \in \cal K$ be the parameter corresponding to the picture $P$ (recall that the width function $n \colon k \mapsto n(k)$ is injective, so $k$ is uniquely defined). We define Layer~0 of $Q$ by embedding the bits of $P_{x,y}$, and then define the content of Layers~1 to~7:
  \begin{enumerate}
  \item Layer 1: we draw a grid of $2k$ columns and $k$ rows of $b(k) \times b(k)$ square cells;
  \item Layer 2: we position the Middle Line in the middle of the picture;
  \item Layer 3 is entirely determined by Layer~1;
  \item Layer 4: We select pairs of blocks
    \(
      (x^{(i_1)}, y^{(i_1)}), \ldots (x^{(i_t)}, y^{(i_t)})
    \)
    with the indices $(i_1,\ldots, i_t)$ taken from the certificate ensuring $f_k(z_1,\ldots z_k) = 1$, where
    \(
    z_j = g(x^{(j)}, y^{(j)}),\ j=1,\ldots, k,
    \)
    are the values of the inner product of the selected blocks);
  \item Layers 5 and 6 are then entirely determined;
  \item Layer 7: we write the integer $k$ on the first tape of the Turing machine. All the computations of the machine are then deterministic.
  \end{enumerate}
  Thus, we obtain a picture $Q$ that satisfies the local constraints of $L$ and such that $\pi(Q) = P$.

  \smallskip
  It remains to observe that such a $Q$ is uniquely determined since any other choices in the layers of $Q$ would eventually lead to a contradiction with the local constraints of $L$: for example, Layer~3 enforces the position of the Middle Line on Layer 2, and its bold \begin{tikzpicture}[scale=.21]\clip (0,0) rectangle ++(1,1); \draw[very thick,countbcolor!70] (0,.5) -- ++(1,0);\end{tikzpicture} line constraints the nondeterministic choice of $k$ on Layer~7. The most interesting is the non-determinism from the block selection on Layer 4: it is actually entirely determined because the ``canonical'' DNF $\varphi$ representing $f_k$ that is computed on Layer 7 is unambiguous; given $(x,y)$ such that $F_k(x,y) = 1$, there exists exactly one possible block selection $(x^{(i_j)},y^{(i_j)})_{1 \leq j \leq t}$.

  This concludes the proof.
\end{proof}

\section{Conclusion}

The main proof presented in this paper designs a picture language embedding a well-chosen communication problem, and is consequently somewhat technical.
The proof involves several geometric gadgets and computation embeddings of Turing machine computations via their space-time diagrams.
There remains an interesting open problem to find a more natural example of a language in $\urec\setminus\corec$, with a simpler proof of membership in \(\urec\). Another natural problem to consider would be whether $\urec$ is closed under union.

It would also be interesting to investigate the extent to which the method developed in this paper can be adapted to study the properties of sofic shifts over multidimensional lattices (\textit{e.g.}~Weiss' question about the entropies of their finite type covers~\cite{desai2006entropies}).

\bibliography{urec}

\end{document}